\documentclass[conference]{IEEEtran}
\ifCLASSINFOpdf
\else
\fi
\usepackage[numbers,sort&compress]{natbib}
\usepackage{amssymb, amsthm, amsmath}
\usepackage[linesnumbered,ruled,vlined]{algorithm2e}
\usepackage{graphicx}
\def\<{\langle}
\def\>{\rangle}

\newcommand{\comment}[1]{}

\newcommand{\h}{\mathcal{H}}

\newtheorem{theorem}{Theorem}

\newtheorem{example}{Example}

\newcommand{\bra}[1]{\langle #1|}
\newcommand{\ket}[1]{|#1\rangle}

\newcommand{\ip}[2]{\langle #1|#2\rangle}
\def\lh{\ensuremath{\mathcal{L(H)}}}

\def\dh{\ensuremath{\mathcal{D(H)}}}
\newcommand{\op}[2]{|#1\rangle \langle #2|}

\newcommand{\tr}{{\rm tr}}

\newcommand{\Prop}{\mathcal{P}}

\begin{document}
%
\title{Entanglement Verification, with or without tomography}

\author{\IEEEauthorblockN{Nengkun Yu}
\IEEEauthorblockA{Centre for Quantum Software and Information,\\
   Faculty of Engineering and Information Technology, \\ University of
   Technology Sydney, NSW 2007, Australia\\
Email: nengkunyu@gmail.com}
}


%


\maketitle

\begin{abstract}
Multipartite entanglement has been widely regarded as key resources in distributed quantum computing,
for instance, multi-party cryptography, measurement based quantum computing, quantum algorithms.
It also plays a fundamental role in quantum phase transitions, even responsible for transport efficiency in biological systems.

Certifying multipartite entanglement is generally a fundamental task.
Since an $N$ qubit state is parameterized by $4^N-1$ real numbers,
one is interested to design a measurement setup that reveals multipartite entanglement with as little effort as possible,
at least without fully revealing the whole information of the state, the so called ``tomography", which requires exponential energy.

In this paper, we study this problem of certifying entanglement without tomography in the constrain that only single copy measurements can be applied.
This task is formulate as a membership problem related to a dividing quantum state space,
therefore, related to the geometric structure of state space.
We show that universal entanglement detection among all states can never be accomplished without full state tomography.
Moreover, we show that almost all multipartite correlation, include genuine entanglement detection, entanglement depth verification,
requires full state tomography. However, universal entanglement detection among pure states can be much more efficient, even we only allow local measurements. Almost optimal local measurement scheme for detecting pure states entanglement is provided.
\end{abstract}


%

\section{Introduction}
Quantum computing has long seemed like one of those technologies that are 20 years away,
and always will be. But 2017 could be the year that the field sheds its research-only image.

The world leading IT giants Google and Microsoft recently hired a host of leading lights, and have set challenging goals for this year.
Their ambition reflects a broader transition taking place at start-ups and academic research labs alike: to move from pure science towards engineering.

Quantum computing offers the potential of considerable speedup over classical computing for some important problems
such as prime factoring \cite{Sh94} and unsorted database search \cite{Gr97}. To take such advantage,
entanglement, one striking feature of quantum many-body systems, must be provided. With shared entanglement, two or more
parties can be correlated in the way that is much stronger than they
can be in any classical way. Entanglement has been widely studied since it has been proven to be an asset to information
processing and computational tasks. For instance, multipartite entanglement has been used as the central
resource for quantum key distribution in multipartite cryptography \cite{MS08}; it is the initial resource
in measurement based quantum computing \cite{RB01}; it is essential in understanding quantum phase transition \cite{SA11};
arguably, multipartite entanglement even should be responsible for transport efficiency in biological systems \cite{SIFW10}.
However, entanglement is still mysterious to many people including experts due to its complex structure.

To understand multipartite entanglement, reliable techniques for characterising entanglement properties of general quantum states are required.
Therefore, it is a fundamental problem to qualitatively test whether a given state is entangled or not.
The pure state case has been extensively studied and fruitful result has been obtained.
For instance, it is proved that almost all multi-qubit entangled states admit
Hardy-type proofs of non-locality without inequalities or probabilities in \cite{ACY16}.
In the setting of of multiple dishonest parties, it is showed how an agent of a quantum network can perform a distributed verification of a source creating multipartite
Greenberger-Horne-Zeilinger (GHZ) states with minimal resources, which is, nevertheless, resistant
against any number of dishonest parties in \cite{PCW+12}.
However, a complete answer of entanglement detection for general mixed states is still missing so far.
A considerable number of different separability criterions have been discovered, including the famous
Positive Partial Transpose(PPT) criterion \cite{PER96}, and Gurvits discovered it lies in the computational complexity class NP-Hard \cite{GUR04}, a
By borrowing idea from functional analysis, entanglement witnesses has been introduced to detect entanglement \cite{HHH96,TER00}.
A more challenging question is the detection of genuine multipartite entanglement, extensive study has not yielded satisfying results yet.

Entanglement detection problem is naturally fallen into the framework of quantum property testing, or quantum characterization,
verification and validation, where one can test any interested property.

Now in the field of quantum computing, we are with a quantum version of the ``big data" problem:
the data collected from quantum systems generated in our labs are growing exponentially because the parameters
is growing exponentially as the number of qubits grows. For instance, an $n$-qubit state has been created in our lad
as the resource of measurement based quantum computing \cite{RB01},
we want to see whether our preparation is correct.

Quantum property testing can be viewed in different settings:
The first one is that the mathematical description of the quantum state is given, in other words,
the complete information of the quantum object is known.
Another one is that the quantum state is given as black box, where one can access its information by measuring it.
Even in the latter setting, two very different scenarios should be considered, statistical fluctuations or accurate measurements.
In the former one, the measurement outcomes are just bit strings distributed according to the outcome probability,
see \cite{MW13} as an excellent survey.
In the latter case,
the measurements of experiments are accurate in the sense that the average of the measurement outcomes is exactly the probability distribution
corresponding to the measurement.
In this paper, we will focus on the latter one.

The general quantum property testing in our setting can be viewed as follows,
\begin{quote}
{\bf Quantum property testing}\\
Let $\mathcal{Q}$ be the set of quantum quantum states.
A subset $\Prop\subseteq{\mathcal{Q}}$ is called a \emph{property}.
An \emph{quantum property tester} for~$\Prop$ is an algorithm (quantum procedure) that receives a black box as input $x\in \mathcal{Q}$. In the former case, the algorithm accepts; in the latter case, the algorithm rejects.
\end{quote}
A property $\Prop$ is called trivial if $\Prop=\mathcal{Q}$ or $\Prop=\emptyset$.

Reconstructing the mathematical description of the given quantum states is called ``quantum state tomography".
Of course, one can obtain any information about this quantum state via quantum state tomography. However, an $N$ qubit state is parameterized by $4^N-1$ real numbers, therefore, informational complete measurements consist of exponential many observables, which is generally impossible.
Formally, we regared the given quantum object as resource, and the goal of property testing is to test the property by accessing the object as less as possible. Therefore, we can define the sample complexity of the property $\Prop$ be the infimum on the number of access the object among all quantum property tester for~$\Prop$. Notice that, the thing we care mostly is how many times do we need to access the quantum object to accomplish the property testing, not the post processing time of the algorithm.
An optimal algorithm may heavily rely on
collective measurements on many identical copies of given states. This
is not friendly for current experimental technology, as collective measurements are
usually much more difficult to implement than measuring single-copy
ones. We will focus on measurements which only applies on single-copy of quantum state.

By noticing that these problems are decision problems with 1 bit outcome, one might hope to achieve the answers with very small number of measurements,
or at least with something less than an informationally complete set-up.
The bipartite version of this problem has been studied recently.
Indeed, it was recently showed that testing whether a bipartite state is entangled or not requires an
informationally-complete measurement \cite{CaHeKaScTo16,Luetal16,CaHeScTo17}.
In \cite{CaHeScTo17}, various sufficient criteria are given, under which the informationally-incomplete measurements can not
reveal the property for unknown quantum state with certainty.
Compare with bipartite entanglement, entanglement in the multipartite setting turns
out much richer and more delicate to characterize.

In this paper, we are going to study two versions of the multipartite entanglement detection problem: We are giving multipartite quantum states, how
do we universally detect entanglement through physical observables?
In the first version, we do not have any further information of the state other than the state space it lives in.
In other words, it can be any mixed state in that state space.
We show that there is no such procedure which can detect multipartite entanglement without full state tomography among all mixed state.
Actually, we prove the more stronger version: For any property that are invariant under stochastic local operations assisted by classical communication(SLOCC)
requires full state tomography unless it is a trivial property if it contains some positive element but not all of them.

Due to the key role of multipartite entanglement in distributed quantum computing, our results can be interpreted as follows, in distributed quantum computation, one can not verify that whether the shared state is entangled or not without reconstructing the state using exponential measurement.

In the second version, we assume that state is pure, and then we provide an almost optimal quantum procedure to detect multipartite entanglement.
Our algorithm only costs linear number of  ``local" measurements, where ``local" means we only need to implement individual measurement on subsystems.
This is extremely friendly for current available technology.

\textbf{Structure of the Paper.} In Section II, we provide technical preliminaries of the basic quantum mechanics.
In Section III, we give the investigation on entanglement together with examples for illustration.
In Section IV, we show that if we do not have any prior information on the given quantum state,
then detecting its entanglement property requires full state tomography. Actually, almost all SLOCC equivalence property here required full state tomography.
In Section V, we show that if we know that the given quantum state is pure, we provide one non-adaptive scheme and one adaptive scheme to detect entanglement which
are exponential faster than doing full state tomography.
Finally, in Section VII, we offer conclusions and some
highlight open problems.

\section{Preliminaries}
For convenience of the reader, we briefly recall some basic notions
from linear algebra and quantum theory which are needed in this paper.
For more details, we refer to \cite{NC00}.

\subsection{Basic linear algebra}
According to a basic postulate of quantum mechanics, the state space
of an isolated quantum system is a Hilbert space. In this paper, we
only consider finite-dimensional Hilbert
spaces. We briefly recall some basic
notions from Hilbert space theory. We write $\mathbb{C}$ for the set
of complex numbers. For each complex number $c\in \mathbb{C}$,
$c^*$ stands for the conjugate of $c$.
An {\it inner product space} $\h$ is a vector space equipped with an inner
product function $$\langle\cdot|\cdot\rangle:\h\times \h\rightarrow \mathbb{C}$$
such that
\begin{enumerate}
\item
$\langle\psi|\psi\rangle\geq 0$ for any $|\psi\>\in\h$, with
equality if and only if $|\psi\rangle =0$;
\item
$\langle\phi|\psi\rangle=\langle\psi|\phi\rangle^{\ast}$;
\item
$\langle\phi|\sum_i c_i|\psi_i\rangle=
\sum_i c_i\langle\phi|\psi_i\rangle$.
\end{enumerate}
For any vector $|\psi\rangle\in\h$, its
length $|||\psi\rangle||$ is defined to be
$\sqrt{\langle\psi|\psi\rangle}$, and it is said to be {\it normalized} if
$|||\psi\rangle||=1$. Two vectors $|\psi\>$ and $|\phi\>$ are
{\it orthogonal} if $\<\psi|\phi\>=0$. An {\it orthonormal basis} of a Hilbert
space $\h$ is a basis $\{|i\rangle\}$ where each $|i\>$ is
normalized and any pair of them are orthogonal.

Let $\lh$ be the set of linear operators on $\h$. For any $A\in
\lh$, $A$ is {\it Hermitian} if $A^\dag=A$ where
$A^\dag$ is the adjoint operator of $A$ such that
$\<\psi|A^\dag|\phi\>=\<\phi|A|\psi\>^*$ for any
$|\psi\>,|\phi\>\in\h$. The fundamental {\it spectral theorem} states that
the set of all normalized eigenvectors of a Hermitian operator in
$\lh$ constitutes an orthonormal basis for $\h$. That is, there exists
a so-called spectral decomposition for each Hermitian $A$ such that
$$A=\sum_i\lambda_i |i\>\<i|=\sum_{\lambda_{i}\in spec(A)}\lambda_i E_i,$$
where the set $\{|i\>\}$ constitutes an orthonormal basis of $\h$, $spec(A)$ denotes the set of
eigenvalues of $A$,
and $E_i$ is the projector to
the corresponding eigenspace of $\lambda_i$.
A linear operator $A\in \lh$ is {\it unitary} if $A^\dag A=A A^\dag=I_\h$ where $I_\h$ is the
identity operator on $\h$.

The {\it trace} of $A\in\lh$ is defined as $\tr(A)=\sum_i \<i|A|i\>$ for some
given orthonormal basis $\{|i\>\}$ of $\h$. It is worth noting that
trace function is actually independent of the orthonormal basis
selected. It is also easy to check that trace function is linear and
$\tr(AB)=\tr(BA)$ for any operators $A,B\in \lh$.

A matrix $A$ is called semi-definite positive if it is Hermitian and has no negative eigenvalues.
A matrix $A$ is called positive if it is Hermitian and has positive eigenvalues only.
We use $A\geq 0$ and $A>0$ to denote the semi-definite positivity and positivity of $A$, respectively.

$||A||$ stands for the 2-norm of $A\in\lh$ by definition $||A||=\sqrt{\tr(A^{\dag}A)}$.

We use $I_{\h}$ to denote the identity operator of $\lh$.

\subsection{Basic quantum mechanics}

According to von Neumann's formalism of quantum mechanics
\cite{vN55}, an isolated physical system is associated with a
Hilbert space which is called the {\it state space} of the system. A {\it pure state} of a
quantum system is a normalized vector in its state space, and a
{\it mixed state} is represented by a density operator on the state
space. Here a density operator $\rho$ on Hilbert space $\h$ is a
semi-definite positive linear operator such that $\tr(\rho)= 1$.
Another
equivalent representation of density operator is probabilistic
ensemble of pure states. In particular, given an ensemble
$\{(p_i,|\psi_i\rangle)\}$ where $p_i \geq 0$, $\sum_{i}p_i=1$,
and $|\psi_i\rangle$ are pure states, then
$\rho=\sum_{i}p_i\op{\psi_i}{\psi_i}$ is a density
operator.  Conversely, each density operator can be generated by an
ensemble of pure states in this way.  The set of
density operators on $\h$ is defined as
$$\mathcal{D}=\{\ \rho\in\lh\ :\rho\ \mbox{is semi-definite positive and} \tr(\rho)=
\mbox{1}\}.$$

The general evolution of a quantum system is described by a trace-preserving super-operator on its state space: if the states of the system at times
$t_1$ and $t_2$ are $\rho_1$ and $\rho_2$, respectively, then
$\rho_2=\sum_k E_k\rho_1E_k^{\dag}$ for some $E_k$.

A (general) quantum {\it measurement} is described by a
Hermitian operator $O$. If the system is in state $\rho$, then the measurement outcome is
$$\tr(O\rho),$$
in the accurate measurement setting of this paper.
\subsection{Tensor Product of Hilbert Space}

The state space of a composed quantum system is the tensor product of the state spaces of its component systems. Let $\h_k$ be a Hilbert space with orthonormal basis $\{|\varphi_{i_k}\}$ for $1\leq k\leq n$.
Then the tensor product $\bigotimes_{k=1}^{n}\h_k$ is defined to be the Hilbert space with $\{|\varphi_{i_1}\rangle...|\varphi_{i_n}\rangle\}$ as its orthonormal basis.
Here the tensor product of two vectors is defined by a new
vector such that
$$\bigotimes_{k=1}^n\left(\sum_{i_k} \lambda_{i_k} |\psi_{i_k}\>\right)=\sum_{i_1,\cdots,i_n} \lambda_{i_1}\cdots\lambda_{i_n}
|\psi_{i_1}\>\otimes\cdots\otimes |\phi_{i_n}\>.$$ Then $\bigotimes_{k=1}^{n}\h_k$ is also a
Hilbert space where the inner product is defined as the following:
for any $|\psi_{k}\>,|\phi_{k}\>\in\h_k$
$$\<\psi_1\otimes\cdots\otimes \psi_n|\phi_1\otimes\cdots\otimes\phi_n\>=\<\psi_1|\phi_1\>_{\h_1}\cdots\<
\psi_n|\phi_n\>_{\h_n}$$ where $\<\cdot|\cdot\>_{\h_k}$ is the inner
product of $\h_k$.

In the bipartite case, the {\it partial trace} of $A\in\mathcal{L}(\h_1
\otimes \h_2)$ with respect to $\h_1$ is defined as
$\tr_{\h_1}(A)=\sum_i \<i|A|i\>$ where $\{|i\>\}$ is an orthonormal
basis of $\h_1$. Similarly, we can define the partial trace of $A$
with respect to $\h_2$. Partial trace functions are also
independent of the orthonormal basis selected.

For a mixed state $\rho$ on $\h_1 \otimes \h_2$,
partial traces of $\rho$ have explicit physical meanings: the
density operators $\tr_{\h_1}\rho$ and $\tr_{\h_2}\rho$ are exactly
the reduced quantum states of $\rho$ on the second and the first
component system, respectively.

\section{Entanglement}
In this section, we introduce some basic facts about the most important quantum feature---Entanglement.

Note that in general, the state of a
composite system cannot be decomposed into tensor product of the
reduced states on its component systems. A well-known example is the
 2-qubit state
$$|\Psi\>=\frac{1}{\sqrt{2}}(|00\>+|11\>).
$$
This kind of state is called {\it entangled state}.
To see the strangeness of entanglement, suppose a measurement $M_{0}=\op{0}{0}$ and $M_{1}=\op{1}{1}$
are applied on the first qubit
of $|\Psi\>$ (see the following for the definition of
quantum measurements). Then after the measurement, the second qubit will
definitely collapse into state $|0\>$ or $|1\>$ depending on whether
the outcome $\lambda_0$ or $\lambda_1$ is observed. In other words,
the measurement on the first qubit changes the state of the second
qubit in some way. This is an outstanding feature of quantum mechanics
which has no counterpart in classical world, and is the key to many
quantum information processing tasks  such as teleportation
\cite{BB93} and superdense coding \cite{BW92}.

In bipartite system, a pure state $\ket{\psi}$ is called product (or not entangled) if it is of form
$$\ket{\psi}=\ket{\psi_1}\ket{\psi_2}.$$
A density matrix $\rho$ is called separable (or not entangled) if it can be written as some convex combination of the density of product pure states, that is $p_i>0$ and semi-definite positive
$\rho_{i,1}$s and $\rho_{i,2}$s such that
$$\rho=\sum_i p_i\rho_{i,1}\otimes \rho_{i,2}.$$
Otherwise, it is called entangled.

An $n$-particle pure state $\ket{\psi}$ is called product if it is of form $$\ket{\psi}=\ket{\psi_1}\cdots\ket{\psi_n}.$$
A density matrix $\rho$ is called separable if it can be written as some convex combination of the density of product pure states. Otherwise, it is called entangled.

\subsection{Positive Partial Transpose}
A bipartite quantum state $\rho\in \mathcal{L}(\h_1\otimes \h_2)$ is
called to have positive partial transpose (or simply PPT) if $\rho^{\Gamma_{\h_1}}\geq 0$, where ${\Gamma_{\h_1}}$ means the partial transpose with respect to the party
$\h_1$, i.e.,
\begin{equation*}
(\op{ij}{kl})^{\Gamma_{\h_1}}=\op{kj}{il}.
 \end{equation*}

This definition can be seen more clearly if we write the state as a block matrix:

The result is independent of the party that was transposed, because
\[
\rho={\begin{bmatrix}A_{11}&A_{12}&A_{13}&\cdots &A_{1n}\\A_{21}&A_{22} &A_{23}&\cdots &A_{2n}\\A_{31}&A_{32}&A_{33}&\cdots &A_{3n}\\ \vdots &\vdots &\vdots &\cdots &\vdots\\ A_{n1}&A_{n2}&A_{n3}&\cdots &A_{nn}\\ \end{bmatrix}}
\]
Where $n$ equals the dimension of $\h_a$, and each block is a square matrix of dimension equals the dimension of $\h_2$. Then the partial transpose is
\[
\rho^{\Gamma_2}={\begin{bmatrix}A^T_{11}&A^T_{12}&A^T_{13}&\cdots &A^T_{1n}\\A^T_{21}&A^T_{22} &A^T_{23}&\cdots &A^T_{2n}\\A^T_{31}&A^T_{32}&A^T_{33}&\cdots &A^T_{3n}\\ \vdots &\vdots &\vdots &\cdots &\vdots\\ A^T_{n1}&A^T_{n2}&A^T_{n3}&\cdots &A^T_{nn}\\ \end{bmatrix}}
\]
It had been observed by Peres that any separable state has positive partial transpose \cite{PER96},
\[
\rho=\sum_i p_i\rho_{i,1}\otimes \rho_{i,2}\\
\Rightarrow \rho^{\Gamma_2}=\sum_i p_i\rho_{i,1}\otimes (\rho_{i,2})^T\geq 0.
\]
The result is independent of the party that was transposed, because $\rho^{\Gamma_1}=(\rho^{\Gamma_2})^T$.

In \cite{KCKL00}, it was proved that all $2\otimes n$ density operators that remain invariant after partial transposition with respect to the first system are separable.
\subsection{Example}

Notice that, a multipartite pure state is product if and only if it is product under any bipartition. However, this is not true for mixed state.

Before going to further introduction on multipartite entanglement,
we first give one example to illustrate the significant difference and complex of multipartite entanglement and bipartite entanglement.

Define three-qubit state as $$\rho=\frac{1}{4}(I-\sum_{i=1^4}\op{\phi_i}{\phi_i}),$$
where $\ket{\phi_i}$ are defined as
\begin{align*}
\ket{\phi_1}=\ket{0,1,+},\\
\ket{\phi_2}=\ket{1,+,0},\\
\ket{\phi_3}=\ket{+,0,1},\\
\ket{\phi_4}=\ket{-,-,-},
\end{align*}
with $\ket{\pm}=\frac{\ket{0}\pm\ket{1}}{\sqrt{2}}$.

One can verify that:

i). $\rho$ is invariant under partial transpose of the any qubit.

Therefore, according to the result just been mentioned of \cite{KCKL00}, $\rho$ is separable in any bipartition.

ii). There is no product state $\ket{\psi_1}\ket{\psi_2}\ket{\psi_3}$ which is orthogonal to all $\ket{\phi_i}$. That is, no product state lives in the orthogonal complement of
the space spanned by $\ket{\phi_i}$.

Notice that $\rho$ is proportional to the projection on the orthogonal complement of
the space spanned by $\ket{\phi_i}$. Therefore, $\rho$ is entangled as it can never be written as the convex combination of the density matrix of product states.

We have constructed three partite entangled state which has no bipartite entanglement.

In other words, multipartite entanglement enjoys much richer structure rather than union of bipartite entanglement.

Remark: The example we constructed here is called unextendable product bases (UPB) investigated in \cite{BDM+99}.

\subsection{Genuine entanglement}

An $n$-particle pure state $\ket{\psi}$ is called genuine entangled if it is not a product state of any bipartition.
To defined the genuine entangle for mixed states, there are two inequivalent ways:

i). A density matrix $\rho$ is called genuine entangled if for any fixed bipartition, it can not be written as some convex combination of the density of pure states which is product in this bipartition.

ii). A density matrix $\rho$ is called genuine entangled if for it can not be written as some convex combination of the density of pure states which is product for any bipartition.

The second definition is stronger than the first one as the bipartition for different pure state can be different.

\subsection{Entanglement depth}
In \cite{SM01}, entanglement depth is introduced to characterize the minimal number of particles in a system that are mutually entangled.

In an $n$-particle system $\h=\h_1\otimes \h_2\otimes\cdots \h_n$,
$\ket{\psi}$ is called $k$-product (separable) if it can be written as
\begin{equation*}
\ket{\psi}=\ket{\psi^{(1)}}\otimes\ket{\psi^{(2)}}\cdots\otimes\ket{\psi^{(k)}},
\end{equation*}
where decomposition corresponds to a partition of the $n$ particles, $\ket{\psi^{(i)}}$ is a genuine entangled state in
$\h_{S_i}=\otimes_{j\in S_i} \h_j$ with $\bigcup S_i=\{1,2,\cdots,n\}$ and $S_i\bigcap S_k=\emptyset$ for $i\neq k$.
The entanglement depth of $\ket{\psi}$, $\mathcal{D}(\psi)$, is defined as the largest cardinality of $S_i$.

An $n$-particle density matrix $\rho$ is called
$k$-separable if it can be written as some convex combination of $k$-separable pure states. The entanglement depth of $\rho_N$ is defined as following,
\begin{equation*}
\mathcal{D}(\rho)=\min_{\rho_N=\sum p_i \ket{\psi^i}\bra{\psi^i}}\max_i~~ \mathcal{D}(\psi^i),
\end{equation*}
where each $\ket{\psi^i}$ is an $N$-particle pure state and $\mathcal{D}(\psi^i)$ is the entanglement depth of $\ket{\psi^i}\bra{\psi^i}$.

\section{Mixed State Property Testing}

In this section, we study the possibility of detecting multipartite correlation
without full state tomography by measuring only single-copy
observables. For simplicity, we allow single-copy
observables are only allowed to be measured nonaddaptively.

We assume that the state is mixed state, and the only know information about this state is the Hilbert space it lives in.
We want to test properties of mixed states. In particular, we are interested in multipartite correlations, SLOCC invariant properties.

Let $\h=\bigotimes_{k=1}^{n}\h_k$ with $d_k$ being the dimension of $\h_k$.
The set (state space) of
density operators on $\h$ is defined as
$$\mathcal{D}=\{\ \rho\in\lh\ :\rho\ \mbox{is semi-definite positive and} \tr(\rho)=
\mbox{1}\}.$$

The concept of stochastic local operations assisted by classical communication (SLOCC) has been used to study entanglement classification
\cite{DVC00,GW13} and entanglement transformation \cite{YCGD10,YGD14}.
Two $n$-partite quantum states $\rho$ and $\sigma$ are called SLOCC equivalent if
\begin{equation*}
\rho=(A_1\otimes A_2\otimes\cdots A_n)\sigma(A_1\otimes A_2\otimes\cdots A_n)^{\dag}
\end{equation*}
holds for some non-singular $A_i\in\lh_k$.

A property $\mathcal{P}\in \mathcal{D}$ is called SLOCC invariant if $\rho\in\mathcal{P}$
implies $\rho'\in\mathcal{P}$ for all $\rho'$ being SLOCC equivalent to $\rho$.

Our main result is given as follows,

\begin{theorem}
For any stochastic local operations assisted by classical communication (SLOCC) invariant property $\mathcal{P}\subsetneq \mathcal{D}$,
such that both $\mathcal{P}$ and $\mathcal{D}\setminus \mathcal{P}$ contain some positive elements respectively,
it is impossible to determine with certainty of whether $\rho\in P$ or not without fully state tomography.

In other words, $t:=\Pi_{i=1}^n d_i^2-1$ measurements are necessary.
\end{theorem}

More precisely, for any set of
informationally-incomplete measurements, there always exists two
different states, $\rho\in\mathcal{P}$ and a
$\sigma\notin \mathcal{P}$, which are not distinguishable according to the measurement results.
That is, for any set of observables (Hermitian matrices) $\{O_1,O_2,\cdots,O_{s}\}$ of $\h$ with $s<t$, there always exist two
different states, $\rho\in\mathcal{P}$ and a
$\sigma\notin \mathcal{P}$, such that $\tr(O_i\rho)=\tr(O_i\sigma)$ for all $i$.

Geometrically, any open and SLOCC invariant nontrivial $\mathcal{P}$ is not `cylinder-like'. In other words,
the structural relation of $\mathcal{P}$ and $\mathcal{D}$ can not as (b).
\begin{figure}[htb]
  \label{fig:main}
  \centering
  \includegraphics[width=\columnwidth]{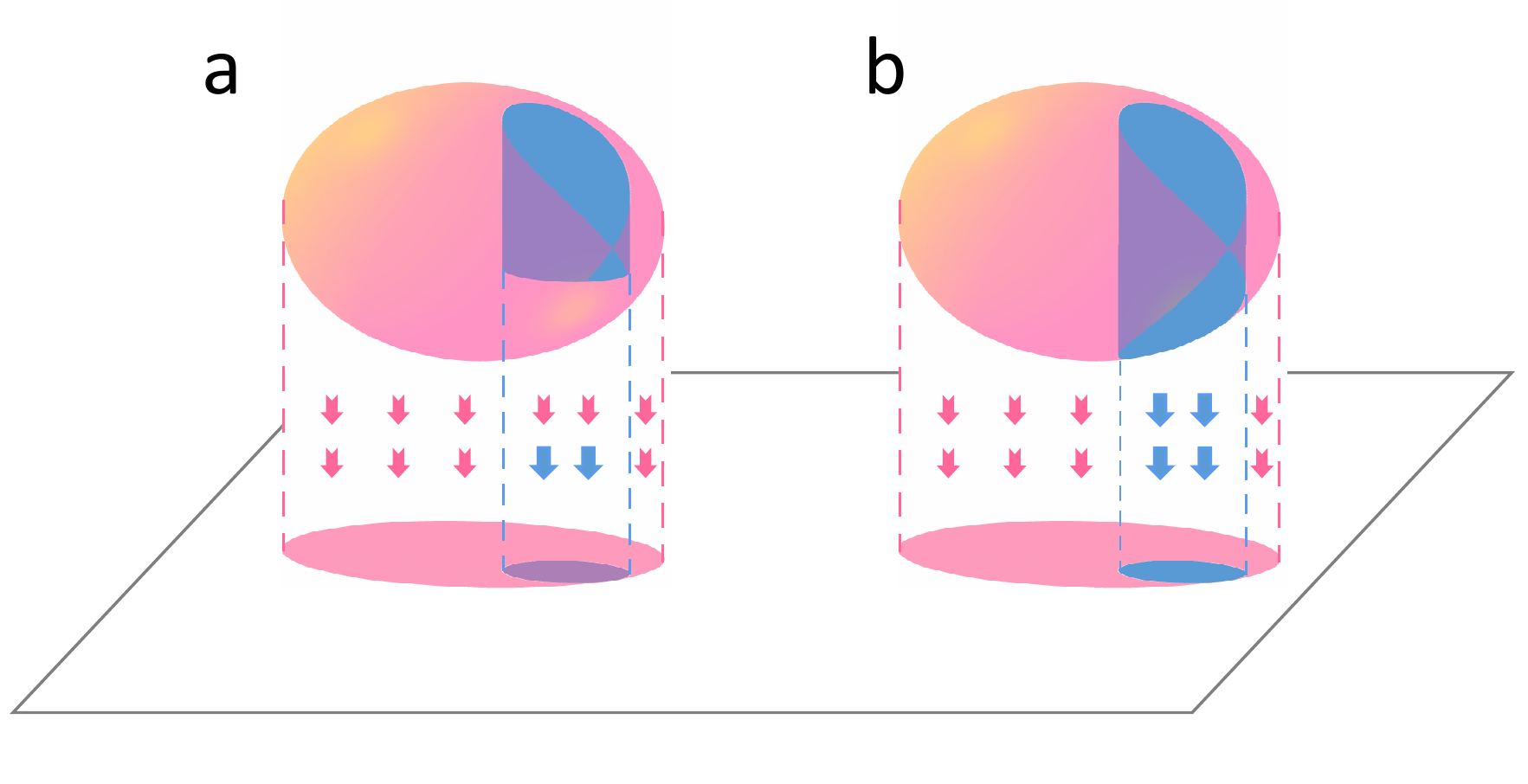}
  \setlength{\abovecaptionskip}{-0.00cm}
  \caption{\footnotesize{Geometry of property $\mathcal{P}$ and its complement. The top pink oval represents
      the set of all states, denoted by $\mathcal{A}$. (a) A set (blue) inside $\mathcal{A}$ which is not
      `cylinder-like'. Hence no projection onto any plane exists
      that can separate the set with the rest states. (b) A set (blue) that
      is an intersection of a generalized cylinder with $\mathcal{A}$ (i.e.
      `cylinder-like'). The projection onto the plane that is orthogonal to the boundary
      lines of the cylinder separates this set with the rest of the states. The bottom
      ovals are the images of the top sets onto a plane, which clear
      show a separation of the images of the blue set from the pink set in (b),
      but an overlap of images in (a).}}  \label{fig:theory}
\end{figure}

\begin{proof}
Notice that, for quantum state in $\h=\bigotimes_{k=1}^{n}\h_k$ with $d_k$ being the dimension of $\h_k$,
the informationally-complete measurements are set of linear independent Hermitian matrices
$\{N_1,N_2,\cdots,N_{t}\}$ as quantum states are trace one which reduces one dimension.

To prove the validity of this theorem, we assume the existence of $\{O_1,O_2,\cdots,O_{s}\}$ of $\h$ with $s<t$ such that for any pairs of $\rho$ and $\sigma$,
one can conclude that $\rho,\sigma\in\mathcal{P}$ or $\rho,\sigma\notin\mathcal{P}$ by giving $\tr(O_i\rho)=\tr(O_i\sigma)$ for all $i$.

The proof is divided into two steps.

STEP 1: We transfer the problem into the existence of informationally-incomplete measurements in testing properties of semi-definite positive operators.

We first generalize the property into all semi-definite operators on $\h$
$$\widetilde{\mathcal{D}}=\{M\in\lh\ :\ M~\mbox{is semi-definite positive}.\}$$

For any property of $\mathcal{D}$, denoted by $\mathcal{P}$, satisfies that $\mathcal{P}\subsetneq \dh$, we first generalize it into
property $\widetilde{\mathcal{P}}$ of $\widetilde{\mathcal{D}}$ as follows,
$$\widetilde{\mathcal{P}}=\{M\in\lh :M/\tr(M)\in\mathcal{P}, A~\mbox{is semi-definite positive}\}.$$

We observe that $\mathcal{P}$ is SLOCC invariant if and only if $\widetilde{\mathcal{P}}$ is SLOCC invariant in the sense that for all non-singular matrices $A_i$s,
$$
M\in \widetilde{\mathcal{P}}\Leftrightarrow (A_1\otimes A_2\otimes\cdots A_n)M(A_1\otimes A_2\otimes\cdots A_n)^{\dag}\in \widetilde{\mathcal{P}}.
$$

$\mathcal{P}$ contains some positive element if and only if $\widetilde{\mathcal{P}}$ is contains some positive element.
$\mathcal{D}\setminus\mathcal{P}$ contains some positive element if and only if $\widetilde{\mathcal{D}}\setminus\widetilde{\mathcal{P}}$ is contains some positive element.

More importantly, one can use the following set of observables
$\{O_0,O_1,O_2,\cdots,O_{s}\}$ with $O_0=I_{\h}$ to test $\widetilde{\mathcal{P}}$ of of $\widetilde{\mathcal{D}}$. Notice that
for any $\widetilde{\rho}\in \widetilde{D}$, we know that $\widetilde{\rho}\in \widetilde{\mathcal{P}}$ if and only if
$\rho\in \mathcal{P}$ with $\rho=\widetilde{\rho}/\tr(O_0\widetilde{\rho})$.
For $0<i$, $\tr(O_i\rho)=\tr(O_i\widetilde{\rho})/\tr(O_0\widetilde{\rho})$.
Thus, for an unknown $\widetilde{\rho}$ with $(o_0,o_1,\cdots,o_s)$ such that $o_i=\tr(O_i\widetilde{\rho})$,
one can conclude that $\widetilde{\rho}\in\widetilde{\mathcal{P}}$ if and only if
the quantum states (trace 1) corresponding to $(o_1/o_0,\cdots,o_s/o_0)$ are in $\mathcal{P}$.
On the other hand, $\{O_0,O_1,O_2,\cdots,O_{s}\}$ is not informationally-complete observalbes as $s<\Pi_{i=1}^n d_i^2-1$.
This indicates informationally-incomplete measurements which can detect $\widetilde{\mathcal{P}}$ with certainty among all $\widetilde{\mathcal{D}}$.

STEP 2: We show there is no informationally-incomplete measurements, $\{O_0,O_1,O_2,\cdots,O_{s}\}$,
which can detect properties of semi-definite positive operators with certainty.

Notice that, there exists an Hermitian $H\neq 0$ such that $\tr(O_iH)=0$ for all $0\leq i\leq s$.
This $H$ enjoys the following property which we called ``free".
For any $\widetilde{\rho}\in \widetilde{\mathcal{D}}$,
if $\widetilde{\rho}+rH\in \widetilde{\mathcal{D}}$ for some $r\in \mathbb{R}$,
then $\widetilde{\rho}+rH\in\widetilde{\mathcal{P}}$ iff $\widetilde{\rho}\in \widetilde{\mathcal{P}}$.

Since $\widetilde{\mathcal{P}}$ is SLOCC invariant, we can conclude that for any ``free" Hermitian $J$, any non-singular $A_i$s,
$$(A_1\otimes A_2\otimes\cdots A_n)J(A_1\otimes A_2\otimes\cdots A_n)^{\dag}$$ is also a ``free" Hermitian by the following observation.
Notice that $\widetilde{\mathcal{D}}$ is SLOCC in variant, then
for any $\widetilde{\rho}$ with $M=\widetilde{\rho}+r(A_1\otimes A_2\otimes\cdots A_n)J(A_1\otimes A_2\otimes\cdots A_n)^{\dag}\in \widetilde{\mathcal{D}}$,
$M\in\widetilde{\mathcal{P}}$
if and only if
$$(A_1^{-1}\otimes A_2^{-1}\otimes\cdots A_n^{-1})M(A_1^{-1}\otimes A_2^{-1}\otimes\cdots A_n^{-1})^{\dag}\in\widetilde{\mathcal{P}}.$$
That is $$(A_1^{-1}\otimes A_2^{-1}\otimes\cdots A_n^{-1})\widetilde{\rho}(A_1^{-1}\otimes A_2^{-1}\otimes\cdots A_n^{-1})^{\dag}+J\in\widetilde{\mathcal{P}}.$$
Invoking the fact that $J$ is ``free", this is equivalent to
$$(A_1^{-1}\otimes A_2^{-1}\otimes\cdots A_n^{-1})\widetilde{\rho}(A_1^{-1}\otimes A_2^{-1}\otimes\cdots A_n^{-1})^{\dag}\in\widetilde{\mathcal{P}}.$$
As $\widetilde{\mathcal{D}}$ being SLOCC invariant, this is true if and only if
$$\widetilde{\rho}\in\widetilde{\mathcal{P}}.$$

This above argument leads us to the fact that $(A_1\otimes A_2\otimes\cdots A_n)J(A_1\otimes A_2\otimes\cdots A_n)^{\dag}$ is also a ``free" Hermitian.

STEP 2 (a): In this part, we show that if the set of ``free" Hermitian matrices is not empty, it contains elements which form a basis of the whole space $\lh$.
In other words, there exist linear independent ``free" Hermitian matrices $H_1,H_2,\cdots,H_t$.

For any nonzero Hermitian $J$, $S=\lh$ where $S$ is the matrix space spanning by all
$$(A_1\otimes A_2\otimes\cdots A_n)J(A_1\otimes A_2\otimes\cdots A_n)^{\dag}$$ with $A_i$s being non-singular.

We actually prove a more general statment.
The set of the following linear maps
$$(A_1\otimes A_2\otimes\cdots A_n)\cdot(A_1\otimes A_2\otimes\cdots A_n)^{\dag}:\lh\rightarrow\lh,$$
spans the whole set of linear maps $\mathcal{L}(\lh,\lh):\lh\rightarrow\lh$.

We start from studying the case $n=1$. In this case, we are going to show that the following maps
$$M\cdot I,I\cdot M$$
lie in the span of
$$A\cdot A^{\dag}$$
with $A$ being non-singular.

Choose real $y\neq 0$ such that $M+yI,M-yI,M-iyI$ being non-singular.
It is easy to verify that $A\cdot I$ lies in the span of
\begin{eqnarray*}
(M+yI)\cdot (M+yI)^{\dag},\\
(M-yI)\cdot (M-yI)^{\dag},\\
(M-iyI)\cdot (M-iyI)^{\dag}.
\end{eqnarray*}
That is
\begin{eqnarray*}
M\cdot I=\frac{1-i}{4y}(M+yI)\cdot (M+yI)^{\dag}\\-\frac{1+i}{4y}(M-yI)\cdot (M-yI)^{\dag}\\-\frac{i}{2y}(M-iyI)\cdot (M-iyI)^{\dag},\\
I\cdot M=\frac{1+i}{4y}(M+yI)\cdot (M+yI)^{\dag}\\-\frac{1-i}{4y}(M-yI)\cdot (M-yI)^{\dag}\\+\frac{i}{2y}(M-iyI)\cdot (M-iyI)^{\dag}.\\
\end{eqnarray*}

One crucial observation is that
For non-singular $A,B$, $AB$ is still non-singular. Thus if any maps $\mathcal{E},\mathcal{F}$ lies in the span of non-singular
$$A\cdot A^{\dag},$$ their composition $\mathcal{E}\circ\mathcal{F}$ also lies in that span.

Therefore, for all $M,N$, we can first implement $M\cdot I$, then apply $I\cdot N$. This observation indicates that
$M\cdot N$
lie in the span of
$A\cdot A^{\dag}$
with $A$ being non-singular.

Notice that any linear maps from $\lh$ to $\lh$ can be written as linear combination of form $M\cdot N$. Thus, for the case $n=1$, the following linear maps
$A\cdot A^{\dag}$
spans the whole set of linear maps $\mathcal{L}(\lh,\lh):\lh\rightarrow\lh$.

Now back to the general $n$ case. Notice that any linear maps from $\lh$ to $\lh$ can be written as linear combination of form $M\cdot N$.
$M\cdot N$ can be written into form $$\sum (A_{1p}\otimes A_{2p}\otimes\cdots A_{np})\cdot(B_{1p}\otimes B_{2p}\otimes\cdots B_{np})$$
with $A_{ip},B_{ip}$ being non-singular matrix of $\h_i$. We can first implement $A_{ip}\cdot B_{ip}$, then tensor them together.
By linearity, we show that the set of the linear maps
$$(A_1\otimes A_2\otimes\cdots A_n)\cdot(A_1\otimes A_2\otimes\cdots A_n)^{\dag}:\lh\rightarrow\lh,$$
spans the whole set of linear maps $\mathcal{L}(\lh,\lh):\lh\rightarrow\lh$.

Therefore, for any nonzero Hermitian $J$, $$(A_1\otimes A_2\otimes\cdots A_n)J(A_1\otimes A_2\otimes\cdots A_n)^{\dag}$$ forms a basis of $\lh$.

STEP 2 (b): In this part, we suppose $H_1,H_2,\cdots,H_t$ with $t=\Pi_{i=1}^n d_i^2-1$ be a set of linear independent ``free" Hermitian matrices. We use the notation
$||\cdot||$ to denote the two norm of the matrix.

We first let $\widetilde{H_1},\widetilde{H_2},\cdots,\widetilde{H_t}$ be the dual basis of $H_1,H_2,\cdots,H_t$.
That is, $\tr(\widetilde{H_i}H_j)=\delta_{i,j}$ for $1\leq i,j\leq t$.

For Hermitian $Y=\sum_{i=1}^t \mu_i H_i$ such that $||Y||=1$, we have $\mu_i=\tr(Y\widetilde{H_i})$. Therefore, $|\mu_i|\leq \sqrt{\tr{\widetilde{H_i}^2}}$.

Now we consider the $t$ matrices $Y_k=\sum_{i=1}^k \mu_i H_i)>0$ for $1\leq k\leq t$. Let $q=\max_{k=1}^t\{\lambda_k+\nu_k:1\leq k\leq t\}$
where $\lambda_k$ denotes the maximal eigenvalue of $Y_k$ and $\nu_k$ denotes the absolution of the minimal eigenvalue of $Y_k$.

For given $\widetilde{\rho}>0$ with $a>0$ being its minimal eigenvalue, we choose $r=\frac{a}{2q}$, then
for any real number $r'$ with $|r'|<r$, and any real numbers $\mu_1,\mu_2,\cdots,\mu_t$ with $||\sum_{i=1}^t \mu_i H_i||=1$,
we have $\widetilde{\rho}+r'(\sum_{i=1}^k \mu_i H_i)>0$ for all $1\leq k\leq t$.

Now we can see that if $\widetilde{\rho}\in \widetilde{\mathcal{P}}$, then for any $M$ with $||\widetilde{\rho}-M||<r$,
we can have $M\in \widetilde{\mathcal{P}}$ by the following argument. Write $M=\widetilde{\rho}+r'Y$ with $Y=\sum_{i=1}^t \mu_i H_i$ and $||Y||=1$, then $r'<r$.
Therefore, $\widetilde{\rho}+r'(\sum_{i=1}^k \mu_i H_i)>0$ for all $1\leq k\leq t$. As $H_1$ is ``free", and $\widetilde{\rho}+r'\mu_1 H_1>0$, we have
$\widetilde{\rho}+r'\mu_1 H_1\in\widetilde{\mathcal{P}}$,$\cdots$, $M=\widetilde{\rho}+r'Y \in \widetilde{\mathcal{P}}$.

If $\widetilde{\rho}\in \widetilde{\mathcal{D}}\setminus\widetilde{\mathcal{P}}$, then for any $M$ with $||\widetilde{\rho}-M||<r$,
we can have $M\in \widetilde{\mathcal{D}}\setminus\widetilde{\mathcal{P}}$ by the following argument. Write $M=\widetilde{\rho}+r'Y$ with $Y=\sum_{i=1}^t \mu_i H_i$ and $||Y||=1$, then $r'<r$.
Therefore, $\widetilde{\rho}+r'(\sum_{i=1}^k \mu_i H_i)>0$ for all $1\leq k\leq t$. As $H_1$ is ``free", and $\widetilde{\rho}+r'\mu_1 H_1>0$, we have
$\widetilde{\rho}+r'\mu_1 H_1\in\widetilde{\mathcal{D}}\setminus\widetilde{\mathcal{P}}$,$\cdots$, $M=\widetilde{\rho}+r'Y \in \widetilde{\mathcal{D}}\setminus\widetilde{\mathcal{P}}$.

Now suppose $0<\widetilde{\rho}\in \widetilde{\mathcal{P}}$ and $0<\widetilde{\sigma}\in \widetilde{\mathcal{D}}\setminus\widetilde{\mathcal{P}}$.
The for any $0\leq l\leq 1$, $M_l=l\widetilde{\rho}+(1-l)\widetilde{\sigma}>0$. Let
$$l:=\sup\{l:M_x\in \widetilde{\mathcal{P}}~~\forall x\leq l\}.$$
Notice that there is a ball of center $\widetilde{\rho}$ lying in $\widetilde{\mathcal{P}}$, then $l>0$.
Also there is a ball of center $\widetilde{\sigma}$ lying in $\widetilde{\mathcal{D}}\setminus\widetilde{\mathcal{P}}$, then $l<1$.
Now we consider the object $M_l$. If $M_l\in \widetilde{\mathcal{P}}$, then there is a ball of radius $r>0$ and center $M_l$ lying in $\widetilde{\mathcal{P}}$,
then there is $\tilde{r}>0$ such that $M_x\in \widetilde{\mathcal{P}}~~\forall x\leq l+\tilde{r}$, contradict to the definition of $l$.
Therefore, $M_l\in \widetilde{\mathcal{D}}\setminus\widetilde{\mathcal{P}}$. Then there is a ball $B$ of center $M_l$ lying in
$\widetilde{\mathcal{D}}\setminus\widetilde{\mathcal{P}}$. Notice that
$$\{M_x: x\leq l\}\bigcap B\neq \emptyset.$$
This is not possible as $\{M_x: x\leq l\}\subset \widetilde{\mathcal{P}}$ and $B\subset \widetilde{\mathcal{D}}\setminus\widetilde{\mathcal{P}}$.

Therefore, there is no informationally-incomplete measurements
which can detect of property $\widetilde{\mathcal{P}}$ of $\widetilde{\mathcal{D}}$ with certainty.

Thus, there is no informationally-incomplete measurements
which can detect of property $\mathcal{P}$ of $\mathcal{D}$ with certainty.
\end{proof}
Almost all properties about multipartite correlations we are interested in are SLOCC invariant.
Theorem 1 indicates that for detecting almost any multipartite correlations, fully state tomography is needed.
In other words, exponential measurement resources are necessary.

In the following, we will applying our result on some examples.
\begin{example}
$\mathcal{P}$ is the set of all PPT states, $i.e.$, states with positive partial transpose.

One can verify that $\mathcal{P}$ is SLOCC invariant.
Obviously, $0<I/t\in \mathcal{P}$, and for sufficient small $x>0$, $xI/t+(1-x)\op{\Phi}{\Phi}\in\mathcal{D}\setminus \mathcal{P}$ with $\ket{\Phi}$ being an entangled pure states.

Applying Theorem 1, we know that fully state tomography is necessary to determine with certainty whether an unknown states is PPT or not.
\end{example}
\begin{example}
$\mathcal{P}$ is the set of all entangled states.

Again, we can use the above arguments. One can verify that $P$ is SLOCC invariant.
Also, $0<I/t\in \mathcal{D}\setminus \mathcal{P}$, and for sufficient small $x>0$, $xI/t+(1-x)\op{\Phi}{\Phi}\in\mathcal{P}$ with $\ket{\Phi}$ being an entangled pure states.

Applying Theorem 1, we know that fully state tomography is necessary to determine with certainty whether an unknown states is entangled or not.
\end{example}

\begin{example}
$\mathcal{P}$ is the set of all states whose entanglement depth is $k$.

Clearly, $\mathcal{P}$ is SLOCC invariant.

If $k=1$, $0<I/t\in \mathcal{P}$, and for sufficient small $x>0$, $xI/t+(1-x)\op{\Phi}{\Phi}\in\mathcal{D}\setminus\mathcal{P}$.
Applying Theorem 1, we know that fully state tomography is necessary to determine with certainty whether the entanglement depth of an unknown states is $k$ or not.

If $1<k\leq n$, $0<I/t\in \mathcal{D}\setminus\mathcal{P}$, and for sufficient small $x>0$, $xI/t+(1-x)\op{\Phi}{\Phi}\in\mathcal{P}$
with $\ket{\Phi}$ being an entangled pure states with depth $k$.
Applying Theorem 1, we know that fully state tomography is necessary to determine with certainty whether the entanglement depth of an unknown states is $k$ or not.

If $k>n$, $\mathcal{P}=\emptyset$, no measurement is needed.
\end{example}

\begin{example}
$\mathcal{P}$ is the set of all genuine entangled states (in any definition given in Section III.C).

One can verify that $\mathcal{P}$ is SLOCC invariant.

$0<I/t\in \mathcal{D}\setminus\mathcal{P}$, and for sufficient small $x>0$, $xI/t+(1-x)\op{\Phi}{\Phi}\in\mathcal{P}$
with $\ket{\Phi}$ being an entangled pure states.

Applying Theorem 1, we know that fully state tomography is necessary to determine with certainty whether an unknown state is genuine entangled or not.

\end{example}

\section{Pure State Entanglement Testing}

In this section, we study the possibility of detecting multipartite correlation
without full state tomography by measuring only single-copy
observables. For simplicity, we allow single-copy
observables and we allow adaptive orocedures.
We provide a lower bound together with an adaptive procedure with almost matching upper bound.

We assume that the state is a pure state, and the only know information about this state is the Hilbert space it lives in.
We want to test whether the state is product or entangled.

Let $\h=\bigotimes_{k=1}^{n}\h_k$ with $d_k$ being the dimension of $\h_k$ and $d_1\geq d_2\geq\cdots\geq d_n$.
The set (state space) of pure state on $\h$ is defined as
$$\{\ \ket{\psi}: \ip{\psi}{\psi}=\mbox{1}\}\subset\h.$$

Our problem can be formalized as following£º

Given a pure quantum $\ket{\psi}$, how many ``local" measurements are needed to verify whether it is product, with in form $\otimes_{k=1}^n\ket{\psi_k}$, or not,
where a measurement is called ``local" if it applied only on one system nontrivially, say $\h_1$, or $\h_2$, or $\cdots$, or $\h_n$.

One can observe the following: $\ket{\psi}$ is product if and only if $\psi_{k}$ is a pure state for any $1\leq k\leq n$ with $\psi_k$
denoting the reduced density operator in $\h_k$. In other words, for any $k$, the resulting operator is pure (rank 1) after tracing out all other system except $k$.

We observe the following lower bound.
\begin{theorem}
Any local ``procedures" that can detect whether an $n$-partite pure state of $\h$ is product or not, must accomplish the pure state tomography of at least $n-1$ parties. Furthermore, at least $\sum_{k=2}^n 2(d_k-1)$ observables are necessary to detect product property.
\end{theorem}
\begin{proof}
As we observed, multipartite entanglement detecting corresponds to purity testing of each parties.

For a given $\sigma_k\in\mathcal{D}_k$, detect whether $\sigma_k$ is pure or not, where
$\mathcal{D}_k$ denotes the mixed state space of $\h_k$,
$$\mathcal{D}_k=\{\ \rho_k\in\lh_k\ :\rho_k\geq 0, \tr(\rho)=
\mbox{1}\}.$$
We first observe that purity testing must accomplish the task of pure state tomography.
In other words, for different pure state $\ket{\psi_{k}},\ket{\phi_{k}}\in \h_k$, the purity testing should be able to distinguish them.
Otherwise, by linearity, it can not distinguish $\op{\psi_k}{\psi_k}$ and $1/2\op{\psi_k}{\psi_k}+1/2\op{\phi_k}{\phi_k}$, where the former one is pure and
the later one is not a pure state. The procedure of testing purity can not determine to output $0$ (pure) or $1$ (not pure).

Suppose for parties $\h_1$ and $\h_2$, the procedure does not accomplish the pure state tomography.
In other words, there exist $\ket{\psi_{1}},\ket{\phi_{1}}\in \h_1$ and $\ket{\psi_{2}},\ket{\phi_{2}}\in \h_2$ such that the procedure can not distinguish them. Then there exist a entangled pure bipartite state $\ket{\Omega}_{12}\in\h_1\otimes\h_2$ such that its reduced density matrices $\Omega_1=\lambda\op{\psi_1}{\psi_1}+(1-\lambda)\op{\phi_1}{\phi_1}$, and $\Omega_2=\mu\op{\psi_2}{\psi_2}+(1-\mu)\op{\phi_2}{\phi_2}$ for some $0<\lambda,\mu<1$.
It is equivalent to find $0<\lambda,\mu<1$ such that $\lambda\op{\psi_1}{\psi_1}+(1-\lambda)\op{\phi_1}{\phi_1}$ and $\mu\op{\psi_2}{\psi_2}+(1-\mu)\op{\phi_2}{\phi_2}$
share the eigenvalues. We only need to choose $\lambda$ to be some very small positive number, then the corresponding $\mu$ does exist.
Now, the procedure can not distinguish the following entangled state
$$\ket{\Omega_{12}}\otimes\ket{\psi_3}\otimes\cdots\otimes\ket{\psi_n}$$
and product state
$$\ket{\psi_1}\otimes\ket{\psi_2}\otimes\ket{\psi_3}\otimes\cdots\otimes\ket{\psi_n},$$
contradict to the assumption that the procedure can detect product property.

Therefore, the procedure must accomplish the pure state tomography of at least $n-1$ parties.

Notice that $d$-dimensional pure state tomography requires $2d-2$ observables as $d$-dimensional pure state has $2d-2$ free real parameters.
Thus, at least $\sum_{k=2}^n 2(d_k-1)$ observables are necessary to detect product property.
\end{proof}

For non-adaptive procedure, the lower bound becomes $\sum_{k=2}^n (4d_k-5)$ as the non-adaptive pure state tomography has lower bound \cite{HeMM12}.

Notice that we do not need to accomplish the purity testing for each party since we have the constrain that the whole state is pure. In that sense $n-1$ parties are enough.

In the following, we provide an upper bound of detecting multipartite entanglement by presenting an algorithm.
We suppose subsystem $\h_k$ with orthornormal basis
$\ket{0},\cdots,\ket{d_k-1}$.

\begin{algorithm}
Let the unknown pure state $\ket{\psi}\in\h$ has $\psi_k$ be its reduced density matrix in subsystem $\h_k$\;
The algorithm output $0$ if $\ket{\psi}$ is product, $1$ if $\ket{\psi}$ is entangled\;
$k\leftarrow 2$\;
$b\leftarrow 0$\;
\While{$b=0$}
{
 $l\leftarrow0$\;
   \While{$\tr(\op{l}{l}\psi_k)=0$}
   {
   $l\leftarrow l+1$\;
   }
   $\alpha_{l,k}\leftarrow\sqrt{\tr(\op{l}{l}\psi_k)}$\;
   $s\leftarrow \alpha_{l,k}^2$\;
   \For {$j = l+1\to d_k-1$}
   {
   Measure $\psi_k$ by $F_j+G_j=\op{j}{l}+\op{l}{j}$\;
   $x\leftarrow \tr[(F_j+G_j)\psi_k]$\;
   Measure $\psi_k$ by $F_j-G_j=i(\op{j}{l}-\op{l}{j})$\;
   $y\leftarrow\tr[(F_j-G_j)\psi_k]$\;
   $\alpha_{j,k}\leftarrow (x+iy)/(2\alpha_{l,k})$\;
   $s\leftarrow s+|\alpha_{j,k}|^2$\;
   }
   \If{$s\neq 1$}
   {
   $b\leftarrow 1$\;
   }
   \Else
   {
   $k\leftarrow k+1$\;
   }
}
Output $b$\;
\caption{Pure Entanglement Testing}
\end{algorithm}

We have the following result.
\begin{theorem}
Algorithm 1 accomplishes the pure entanglement testing in $\h$ by using at most $\sum_{k=2}^n (2d_k-1)$ observables.
\end{theorem}
\begin{proof}
To prove Algorithm 1 accomplishes the pure entanglement testing in $\h$, we need to show two directions.

One direct is Algorithm 1 output $0$ if $\ket{\psi}$ is product. In other words, $\psi_k$ is pure for any $1\leq k\leq n$.
As $\ket{\psi}$ is pure, we only need to prove that $\psi_k$ is pure for any $2\leq k\leq n$.

Assume
\begin{equation*}
\ket{\psi_k}=\sum_{m=0}^{d_k-1}\beta_{m,k}\ket{m}.
\end{equation*}
According to the protocol, at Line 7, we measure $\ket{\psi_k}$ using measurements $E_l$ sequentially until $\tr(\ket{\psi_k}\bra{\psi_k} E_l)$ is non-zero, where $E_m=\op{m}{m}$. The goal is to find the smallest $l$ such that $\beta_l\neq 0$. The state becomes
\begin{equation*}
\ket{\psi_k}=\sum_{m=l}^{d_k-1}\beta_{m,k}\ket{m},
\end{equation*}
where the summation starts from $m=l$ now.
Now we know that $\alpha_{k,l}=\beta_{l,k}=\sqrt{\tr(\ket{\psi}\bra{\psi} E_k)}$ is positive since the global phase of a quantum state is ignorable.

The goal of Line 12 to Line 16 is to obtain $\beta_{m,j}$ for all $m\geq l$ by employing the coherence between $\ket{m}$ and $\ket{l}$. In terms of density matrix, our protocol actually provides the ($j+1$)-th row of $\op{\psi}{\psi}$.
Now we have
\begin{align*}
x=\langle\psi_k|(F_j+G_j)\ket{\psi_k}&=&\beta_{l,k}\beta_{j,k}+\beta_{l,k}\beta_{j,k}^*, \\ \nonumber
y=\langle\psi_k|(F_j-G_j)\ket{\psi_k}&=&i(\beta_{l,k}\beta_{j,k}-\beta_{l,k}\beta_{j,k})^*.
\end{align*}
As we have assumed that $\beta_{l,k}$ is real, it is obvious that $\beta_{l,k}^*\beta_{j,k}^*=\beta_{l,k}\beta_{j,k}^*$ for all $j>l$. Therefore, we can calculate the exact value of $\beta_{j,k}$ since we know the non-zero $\alpha_k$ and $\beta_{l,k}\beta_{j,k}$ from our measurements.
\begin{align*}
\beta_{j,k}=\frac{x+iy}{2\alpha_{l,k}}=\alpha_{j,k}.
\end{align*}
According to the fact that $\ket{\psi_k}$ is a normalized pure state, we have
\begin{align*}
\sum_{j=l}^{d_k-1}|\alpha_{j,k}|^2=\sum_{j=l}^{d_k-1}|\beta_{j,k}|^2=\ip{\psi_k}{\psi_k}=1.
\end{align*}
Therefore, if all $\psi_k$ are pure state, then Line 18-19 of Algorithm 1 will never be called. That means, the output $b$ is $0$.

In the next, we show the other direction. If $\ket{\psi}$ is entangled, then Algorithm 1 outputs $1$.
To derive a contradiction, we assume that Algorithm 1 outputs $0$ for some entangled $\ket{\psi}$.
We first notice that if $\ket{\psi}$ is entangled, there exist $k>1$ such that $\psi_k$ is not a pure state.
In the next, we suppose there is the smallest $p>1$ such that $\psi_p$ is not a pure state.
According to the previous argument, then the execution of Line 5-21 in Algorithm for such all $1<k<p$, would not change the value of $b$ as $\psi_k$ is pure state here.
If the value of $b$ is not changed after the execution of Line 5-21 in Algorithm for $k=p$, we know that $\sum_{j=l}^{d_k-1}|\alpha_{j,k}|^2=1$.
Therefore, we can define a pure state as follows
\begin{align*}
\ket{\phi_k}=\sum_{m=0}^{d_k-1}\alpha_{m,k}\ket{m}.
\end{align*}
We prove that $\psi_k=(r_{ka,kb})_{d_k\times d_k}$ is pure by showing $\psi_k=\op{\phi_k}{\phi_k}$.

For $m<l$, we have
\begin{align*}
r_{km,km}=\tr(\psi_k\op{m}{m})&=\tr(\op{\phi_k}{\phi_k}\op{m}{m})=0,\\
r_{kl,kl}=\tr(\psi_k\op{l}{l})&=\tr(\op{\phi_k}{\phi_k}\op{l}{l})=\alpha_{k,l}^2.
\end{align*}
For $l\leq m\leq d_k-1$, we have
\begin{align*}
r_{km,kl}=\tr(\psi_k\op{l}{m})&=\tr(\op{\phi_k}{\phi_k}\op{l}{m})=\alpha_{k,l}\alpha_{k,m}.
\end{align*}
As $\psi_k$ is semi-definite positive, we know that the first $l$ rows and columns of $\psi_k$ are all zero.

For any $l\leq m\leq d_k-1$, we choose the sub-matrix of $\psi_k$ of $\{\ket{l},\ket{m}\}\times\{\bra{l},\bra{m}\}$,
\[
{\begin{bmatrix}\alpha_{k,l}^2&\alpha_{k,l}\alpha_{k,m}^*\\ \alpha_{k,l}\alpha_{k,m}&r_{km,km}\\ \end{bmatrix}}
\]
This sub-matrix is also semi-definite positive. Thus,
$$
r_{km,km}\geq |\alpha_{k,m}|^2.
$$
According to $\tr(\psi_k)=1$, we have
$$
1=\sum_m r_{km,km}\geq \sum_m |\alpha_{k,m}|^2=1.
$$
Thus, $r_{km,km}=|\alpha_{k,m}|^2$.

Now for any $m,s>l$, we choose the sub-matrix of $\psi_k$ of $\{\ket{l},\ket{m},\ket{s}\}\times\{\bra{l},\bra{m},\bra{s}\}$,
\[
{\begin{bmatrix}\alpha_{k,l}^2&\alpha_{k,l}\alpha_{k,m}^*&\alpha_{k,l}\alpha_{k,s}^*\\\alpha_{k,l}\alpha_{k,m}&|\alpha_{k,m}|^2 &r_{km,ks}^*
\\ \alpha_{k,l}\alpha_{k,s}&r_{km,ks}&|\alpha_{k,s}|^2\\ \end{bmatrix}}.
\]
According to its positivity of determinant, we have $r_{km,ks}=\alpha_{k,s}\alpha_{k,m}$.
That is $\psi_k=\op{\phi_k}{\phi_k}$.
This contradict to our assumption that $\psi_k$ is not pure.
Therefore, if $\ket{\psi}$ is entangled, Algorithm 1 would output $1$.

For each $k>1$, the execution of testing $\psi_k$ uses at most $2d_k-1$ observables: Line 6-7 uses $l$ observables, Line 11-17 uses $2d_k-l$ observables. In total, ALgorithm 1 uses at most $\sum_{k=2}^n (2d_k-1)$ observables.
\end{proof}

The gap between our upper bound and lower bound is at most $n-1$.

\section{Conclusion}
In this paper, we study this problem of certifying entanglement without tomography in the constrain that only single copy measurements can be applied.
We show that almost all multipartite correlation, include genuine entanglement detection, entanglement depth verification,
requires full state tomography. However, universal entanglement detection among pure states can be much more efficient, even we only allow local measurements. Almost optimal adaptive local measurement scheme for detecting pure states entanglement is provided.

There are still many interesting open problems related to this topic. An immediate one is to generalize Theorem 1.
There are two possible directions, about local unitary invariant property and adaptive measurements.

This work is supported by DE180100156.


\end{document}